\documentclass[conference,twocolumn]{IEEEtran}
\IEEEoverridecommandlockouts

\usepackage{amsmath}
\usepackage{amssymb}

\usepackage{amsfonts}
\usepackage{float} 
\usepackage{srcltx}
\usepackage{mathrsfs} 
\usepackage{multirow}

\usepackage{graphicx}
\usepackage{epsfig}
\usepackage{psfrag}
\usepackage{subfigure}

\usepackage{setspace}
\usepackage{dsfont} 

\usepackage[]{units} 
\usepackage{url} 
\usepackage[dvips]{color}
\usepackage{verbatim} 
\usepackage{cite} 
\usepackage{ifthen} 
\usepackage{ifpdf} 

\newcommand{\cardinality}[1]{\ensuremath{\lvert#1\rvert}}
\newcommand{\EE}{\ensuremath{\mathbb{E}}}

\newcommand{\pr}[1]{\ensuremath{\mathbb{P}\left( #1 \right)}}
\newcommand{\indicator}[1]{\ensuremath{\mathds{1}_{\left\{ #1 \right\}}}}

\newcommand{\mspcc}{\ensuremath{\;}}
\newcommand{\mspcd}{\ensuremath{\quad}}

\newcommand{\mset}[1]{\ensuremath{\mathcal{#1}}}
\newcommand{\mvec}[1]{\ensuremath{\textbf{#1}}}

\newcommand{\mfield}[1]{\ensuremath{\mathbb{#1}}} 

\newcommand{\mDefine}{\ensuremath{ \stackrel{\triangle}{=} }}

\newcommand{\mZZ}{\ensuremath{\mathbb{Z}}}

\newtheorem{lemma}{Lemma}

\newtheorem{theorem}{Theorem}
\newtheorem{definition}{Definition}

\newtheorem{proposition}{Proposition}

\newtheorem{remark}{Remark}

\newcommand{\XX}{\ensuremath{\mvec{X}}}
\newcommand{\xx}{\ensuremath{\mvec{x}}}
\newcommand{\YY}{\ensuremath{\mvec{Y}}}
\newcommand{\yy}{\ensuremath{\mvec{y}}}
\newcommand{\NN}{\ensuremath{\mvec{N}}}
\newcommand{\zz}{\ensuremath{\mvec{z}}}
\newcommand{\DD}{\ensuremath{\mvec{D}}}
\newcommand{\ww}{\ensuremath{\mvec{w}}}

\renewcommand{\mid}{\middle\vert}

\begin{document}

\allowdisplaybreaks

\title{The Importance of Tie-Breaking in Finite-Blocklength Bounds}

\author{
\authorblockN{Eli Haim}
\authorblockA{Dept. of EE-Systems,
TAU\\
Tel Aviv, Israel \\
Email: elih@eng.tau.ac.il
}
\and
\authorblockN{Yuval Kochman}
\authorblockA{School of CSE,
HUJI\\
Jerusalem, Israel \\
Email: yuvalko@cs.huji.ac.il
}
\and
\authorblockN{Uri Erez}
\authorblockA{Dept. of EE-Systems,
TAU\\
Tel Aviv, Israel \\
Email: uri@eng.tau.ac.il
}
}

\maketitle


\begin{abstract}
We consider upper bounds on the error probability in channel coding. We derive an improved maximum-likelihood union bound, which takes into account events where the likelihood of the correct codeword is tied with that of some competitors. We compare this bound to various previous results, both qualitatively and quantitatively. With respect to maximal error probability of linear codes, we
observe that when the channel is additive, the derivation of bounds, as well as the assumptions on the admissible encoder and decoder, simplify considerably.
\end{abstract}

\section{Introduction}

Consider maximum-likelihood decoding, known to be optimal in the sense of average error probability between equi-probable messages. What happens when $\ell$ false codewords share the maximum likelihood score with the transmitted one? No matter how such a tie is broken, the average error probability given this event is $1-\nicefrac{1}{\ell+1}$. Computing the optimal error probability, taking into account all possible ties, is exponentially hard. Can we ignore this event, i.e., assume that in case of a tie the decoder is always right or always wrong? The answer depends upon both the channel and the blocklength. When the likelihood score is a continuous random variable, the probability of ties is zero. Also, for long enough blocks, the distribution of the score of a word can be closely approximated by a continuous one (e.g., using the central-limit theorem). However, for small enough alphabet size and short enough blocks, the effect of ties on error-probability bounds is not negligible.

We revisit the finite-blocklength achievability results of Polyanskiy et al. \cite{PolyanskiyPVFiniteLength10}. For i.i.d. codewords, and when we can neglect ties, computation of the exact average error probability is not harder than that of the random-coding union (RCU) achievability bound. However, ties cannot always be neglected. As the RCU bound assumes that ties always lead to errors, it can be improved; indeed, we derive a tighter bound. In particular, unlike the RCU bound, the new bound is always tighter than bounds based upon threshold decoding.

When it comes to maximal error probability, tie-breaking is no longer a mere issue of analysis. Rather, ties have to be broken in a manner that is ``fair'', such that the error probability given different messages is balanced. In \cite{PolyanskiyPVFiniteLength10}, a randomized decoder is employed in order to facilitate such fairness. But is randomization necessary? We show that at least for additive channels (over a finite field), a deterministic decoder suffices.

\section{Notation and Background}

We consider coding over a memoryless channel with some finite blocklength $n$, i.e.:
\begin{align} V(\yy|\xx) = \prod_{i=1}^n V(y_i|x_i). \end{align}
for every $\xx \in \mset{X}^n, \yy \in \mset{Y}^n$.
The channel input and output alphabets are arbitrary.
For the sake of simplicity, we adopt discrete notation; the bounds do not depend on alphabet sizes, and the educated reader can easily translate the results to the continuous case (which is of limited interest in the context of tie-breaking).
The codebook is given by $\xx_1,\ldots,\xx_M$, where $M$ is the number of codewords (that is, the coding rate is $R=\nicefrac{1}{n} \log M$).
The decoder produces an estimate $\hat m$, where the transmitted message index is denoted by $m$. The average error probability, assuming equiprobable messages, is given by:
\begin{align} \epsilon = \frac{1}{M} \sum_{m=1}^M \pr{\hat m \neq m | \XX=\xx_m}. \end{align}
The maximum error probability is given by
\begin{align} \label{eq:maximal} \overline \epsilon = \max_{m=1\ldots M} \pr{\hat m \neq m | \XX=\xx_m}. \end{align}

For the sake of analyzing the error probability, it is convenient to consider code ensembles. All ensembles we consider in this work fall in the following category.
\begin{definition}[Random conditionally-symmetric ensemble]
\label{def:RCSE}
An ensemble is called random conditionally-symmetric ensemble (RCSE) if its codewords are drawn such that for every different $m,j,k \in \{1,\ldots\,M\}$ and for every $\xx,\bar{\xx} \in \mathcal{X}^n$:
\begin{align}
\pr{ \XX_j=\bar{\xx} \mid \XX_m=\xx } = \pr{ \XX_k=\bar{\xx} \mid \XX_m=\xx }
\end{align}
\end{definition}
It is easy to verify, that for an RCSE, all words are identically distributed. We can thus define by $\XX$ a word drawn under the ensemble distribution (not necessarily memoryless) $P$ over the set $\mset{X}^n$.
Using this input distribution, the information density is given by:
\begin{align} \label{eq:information_density} i(\xx;\yy) = \log \frac {V(\yy|\xx)}{PV(\yy)}, \end{align}
where $PV(\yy)$ is the output distribution induced by $P(\xx)$ and $V(y|x)$. We denote by $\YY$ the output corresponding to the random input $\XX$, and the random variable $i(\XX;\YY)$ is defined accordingly. In addition, we define $i(\bar\XX;\YY)$ as the information density a codeword $\bar\XX$ other\footnote{In a random codebook it may happen that the codebook contains some identical codewords. Thus it is possible that $\bar{\XX}=\XX$, as long as they represent different messages.} than the one that generated $\YY$.\footnote{In \cite{PolyanskiyPVFiniteLength10}, the notation $i(\XX;\bar\YY)$ is sometimes used; for RCSE, the two are equivalent.}

The importance of deriving bounds for an RCSE is due to the fact that this class includes many interesting ensembles. An important special case of RCSE is the pairwise-independent ensemble:
\begin{definition}[Pairwise-independent ensemble]
A pairwise independent ensemble (PIE) is an ensemble such that its codewords are pairwise-independent and identically distributed. That is, for any two indices $i\neq j$, \begin{align} \pr{\XX_i=\xx_i|\XX_j=\xx_j} = \pr{\XX_i=\xx_i} = P(\xx). \end{align}
\end{definition}

We note that the codewords of an RCSE are not necessarily pairwise-independent. One example is linear random codes with a cyclic generating matrix~\cite{Seguin79}. In this ensemble, a codebook is a linear code, such that all the cyclic shifts of the any codeword are also codewords.
Generally, RCSE (which are not necessarily PIE) can be constructed by first drawing a class of codewords, and then, randomly (uniformly) drawing the codewords from this class.
Alternatively, it can be constructed by choosing some codeword which defines the class, from which all the other codewords will be drawn.

Finally, the following class of channels turns out to play a special role.
\begin{definition}[Additive channels] \label{definition:additive} A channel is additive over a finite group $\mset{G}$ with an operation, if $\mset{X}=\mset{Y}=\mset{G}$, and the transition distribution $V(y|x)$ is compatible with
\[ Y = X + N \]
where $N$ is statistically independent of $X$, and ``$+$'' denotes the operation over $\mset{G}$.\footnote{The operation ``$-$'' over the group, which is uniquely defined by the operation ``$+$'', such that for any $a,b,c \in \mset{G}: \mspcc a-b=c \mspcd \textrm{iff} \mspcd a=c+b$.}
\end{definition}
For example, for modulo-additive channels the alphabet is the ring $\mZZ_q$, and addition is modulo $q$. The importance of additive channels stems from the following.
\begin{lemma}
\label{lemma:additive}
Consider an additive channel over $\mset{G}$, and a codebook drawn from a PIE with uniform input distribution over $\mset{G}^n$, i.e. $P(\xx)=\cardinality{\mset{G}}^{-n} \mspcd \forall \xx\in\mset{G}^n$.
Then, $i(\bar\XX;\YY)$ is statistically independent of $(\XX,\YY)$. \end{lemma}
\begin{proof}
For this channel the information density~\eqref{eq:information_density} is equal to
\begin{align}
\label{eq:infromation_density_of_additive_channel}
i(\xx;\yy) = \log\frac{P_\NN(\yy-\xx)}{P_\YY(\yy)},
\end{align}
where $P_\NN(\cdot)$ is the distribution of the noise, and $P_\YY(\cdot)$ is the distribution of the channel output.
For this channel with codebook drawn from a PIE with a uniform distribution over $\mset{G}^n$, we have that for every $\zz \in \mset{G}^n$:
\begin{subequations}
\begin{align}
\pr{\YY-\bar\XX = \zz} &= \pr{\XX+\NN-\bar\XX\ = \zz}
\\&= \cardinality{\mset{G}}^{-n},
\end{align}
\end{subequations}
since $\bar\XX$ is uniformly distributed over $\mset{G}^n$ and statistically independent of $(\XX,\NN)$. Therefore, $\YY-\bar{\XX}$ is statistically independent of $(\XX,\YY)$; Moreover, any function of $\YY-\bar{\XX}$ is also statistically independent of $(\XX,\YY)$, in particular $P_N(\YY-\bar{\XX})$ is statistically independent of $(\XX,\YY)$.

Since $\XX$ is uniformly distributed over $\mset{G}^n$, and is statistically independent noise, then the channel output $\YY$ is also uniformly distributed over $\mset{G}^n$, i.e. for any $\yy \in \mset{G}^n$:
\begin{align}
P_\YY(\yy) &= \cardinality{\mset{G}}^{-n},
\end{align}
and hence, $P_\YY(\YY)$ is statistically independent of $(\XX,\YY)$.
From the two observations above, we conclude that $i(\bar\XX;\YY)$ is statistically independent with $(\XX,\YY)$.
\end{proof}

%
\section{I.I.D. Codebooks}
%

Before stating the main results that apply to any RCSE, we start by simple bounds that hold for the special case of an i.i.d. ensemble. That is, all codewords are mutually independent, and each one distributed according to $P(\XX)$. In this case, the average error probability is well known, although hard to compute \cite{PolyanskiyPVFiniteLength10}. Denote:
\begin{subequations} \label{eq:WZ}
\begin{align}
W & = \pr{ i(\bar{\XX};\YY) = i(\XX;\YY) | \XX, \YY } \\
Z & = \pr{ i(\bar{\XX};\YY) < i(\XX;\YY) | \XX, \YY }.
\end{align}
\end{subequations}
Then, for an i.i.d. ensemble \cite[Thm. 15]{PolyanskiyPVFiniteLength10}:
\begin{align} \label{PPV_exact}
\epsilon^{\textrm{(iid)}} = 1 - \sum_{\ell=0}^{M-1} \frac{1}{\ell+1} \cdot \binom{M-1}{\ell} \EE_{\XX,\YY} \left( W^\ell Z^{M-1-\ell} \right) . \end{align}
This result stems from the fact that for equiprobable words, maximum likelihood (ML) decoding is just maximum information density. We note that $\ell$ represents the number of competing codewords that share the maximal information-density score with the correct one; given $\ell$, the correct codeword will be chosen with probability $\nicefrac{1}{\ell+1}$. If $W=0$ (as happens when $V(Y|x)$ is a proper density for every $x \in \mset{X}$), the calculation is straightforward. Otherwise, it has exponential complexity. Thus, the main burden is with dealing with ties. In order to avoid such burden, we suggest the following simple bounds.

\begin{proposition}[Bounds for i.i.d. codebooks]
\label{prop:iid_bound}
For an i.i.d. ensemble,
\begin{align}
1 - \EE_{\XX,\YY} \left[ (W+Z)^{M-1} \right] \leq \epsilon^{\textrm{(iid)}} \leq 1 - \EE_{\XX,\YY} \left[ Z^{M-1} \right].
\end{align}
\end{proposition}

This result can be shown either from \eqref{PPV_exact} or directly. For the lower bound, in case multiple codewords (including the correct one) attain the maximal information density, the correct one is always chosen; for the upper bound, it is never chosen under such circumstances. Of course, as the upper bound is just the first term in \eqref{PPV_exact}, one may tighten it by taking more terms. The difference between the lower and upper bounds may be quite significant, as demonstrated in Figure~\ref{fig:iid_bounds}.

\begin{figure}
\centering
\includegraphics[width=\columnwidth]{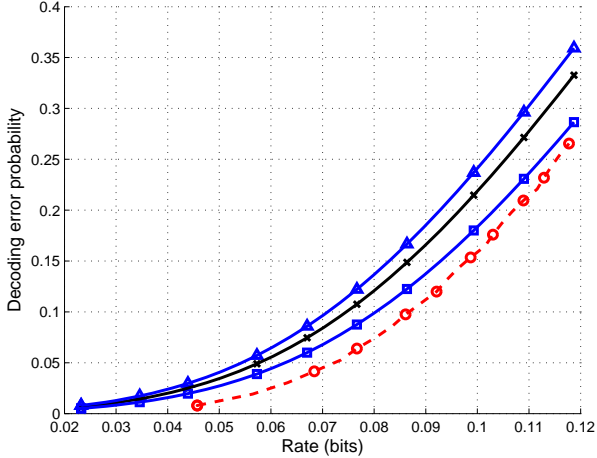} 
\caption{The effect of tie-breaking on the performance of i.i.d. codebooks. We demonstrate the effect using a BSC with crossover probability $0.3$, at blocklength $n=100$. The triangle- and square- marked solid curves give the lower and upper bounds of Proposition~\ref{prop:iid_bound}, respectively. The $\times$-marked solid curve is the exact error probability of the i.i.d. ensemble \eqref{PPV_exact}, evaluated by taking enough terms in the sum, such that the effect of additional ones is numerically insignificant. For reference, the circle-marked dashed curve gives the tightest lower bound on the error probability, which holds for \emph{any} codebook~\cite[Theorem 16]{PolyanskiyPVFiniteLength10}.}
\label{fig:iid_bounds}
\end{figure}

%
\section{Bounds for RCSE}
%
\label{sec:RCU}

\subsection{Maximum-Likelihood Union Bounds}

When the codewords are not i.i.d., we cannot use anymore products of probabilities as done in the previous section. However, for providing a lower bound on the error probability, we can use a union bound. We derive a result that is close in spirit to the RCU bound \cite[Theorem 16]{PolyanskiyPVFiniteLength10}, which states that $\epsilon^{\textrm{(iid)}} \leq \epsilon_\textrm{RCU}$,\footnote{Indeed, it is noted in \cite[Appendix A]{PolyanskiyPVFiniteLength10} that pairwise independence is sufficient.}  where
\begin{align}
\label{eq:RCU}
\epsilon_\textrm{RCU} \triangleq \EE_{\XX,\YY} \left[ \min \left\{
1,
(M-1) \cdot \left(1-Z\right) \right\} \right].
\end{align}
We improve this bound in two ways: First, it is extended to any RCSE, and second, the case of equal maximal information-density scores is taken into account.

\begin{theorem}[RCU$^*$ bound]
The average error probability of an RCSE satisfies $\epsilon^{\textrm{(RCSE)}} \leq \epsilon_{\textrm{RCU}^*}$ where
\begin{align}
\label{eq:RCU_star}
 \epsilon_{\textrm{RCU}^*} \triangleq \EE_{\XX,\YY} \left[ \min \left\{
1,
(M-1) \cdot \left(1-Z-\frac{W}{2}\right) \right\} \right],
\end{align}
where the conditional probabilities $W$ and $Z$ are given by \eqref{eq:WZ}.
\end{theorem}

\begin{proof}
Without loss of generality, assume that the transmitted codeword index is $m=1$. The ML decoder will choose the codeword with maximal information density; in case of equality, it will uniformly draw a winner between the maximal ones. Let $C_j$ be the event that the codeword $j$ was chosen in such a lottery.
Denote the following events:
\begin{subequations}
\begin{align}
A_j &\mDefine \left\{ i(\XX_j;\YY) > i(\XX;\YY)  \right\}
\\B_j &\mDefine \left\{ i(\XX_j;\YY) = i(\XX;\YY)  \right\}.
\end{align}
\end{subequations}
Also denote $A \mDefine \bigcup_{j=2}^M A_j$ and $B \mDefine \bigcup_{j=2}^M B_j$.
Then, the error probability is given by:
\begin{subequations}
\begin{align}
\epsilon^{\textrm(RCSE)} &= \pr{ A \cup \left[ B \cap \bar C_1 \right] \mid m=1 }
\\&=\EE_{\XX_1,\YY} \pr { A \cup \left[ B \cap \bar C_1 \right] \mid m=1,\XX_1,\YY }.
\\&=\EE_{\XX_1,\YY} \min\left\{1, \pr{ A \cup \left[ B \cap \bar C_1 \right] \mid m=1,\XX_1,\YY } \right\}
\\&\mDefine \pr{ A \cup \left[ B \cap \bar C_1 \right] \mid m=1,\XX_1,\YY }. \label{eq:RCU_star:proof:sub_in}
\end{align}
\end{subequations}
Using the union bound between events of equality and inequality, we have:
\begin{align}
S \leq \pr{A \mid m=1,\XX_1,\YY } + \pr{B \cap \bar C_1 \mid m=1,\XX_1,\YY  }.
\end{align}
Now, the event $C_1$ depends on the the rest of the variables only through the number of codewords that achieve equal score. Specifically, if there are $\ell$ impostors, then $\pr{C_1} = \nicefrac{1}{\ell+1}$. Since the second term is non-zero only if $\ell \geq 1$, it follows that:
\begin{align}
S \leq \pr{A \mid m=1,\XX_1,\YY  } + \frac{1}{2} \pr{B \mid m=1,\XX_1,\YY }.
\end{align}
We now use the union bound, as in \cite{PolyanskiyPVFiniteLength10}:
\begin{align}
S &\leq \sum_{j=2}^M \pr{ A_j \mid m=1,\XX_1,\YY  } + \frac{1}{2} \sum_{j=2}^M \pr{  B_j \mid m=1,\XX_1,\YY }.
\end{align}
Noting that for an RCSE each element in the left (resp. right) sum equals $1-W-Z$ (resp. $W$), and substituting this bound in~\eqref{eq:RCU_star:proof:sub_in} we arrive at the desired result.
\end{proof}

\begin{remark} \label{remark:pairwise} We can give the RCU$^*$ bound the following interpretation. First, each potential input $\xx_j$ is given an information-density score (equivalent to a likelihood score) $i_j$. Then, these scores are fed to a comparison process. The process is biased against the correct codeword, in the sense that it has to beat each and every impostor. However, each pairwise comparison itself is optimal (the correct codeword will beat an impostor with lower score), and fair (in case of a tie, both codewords are equally likely to win). This comparison mechanism is worse than the actual decoder used in the proof, since in case the correct codeword shares the maximal score with $\ell$ impostors, it has probability $2^{-\ell}$ to be chosen, rather than $\nicefrac{1}{\ell+1}$; yet, the union bound for both is equal. \end{remark}

\subsection{Relation to Gallager's Type-I bound}
\label{sec:Gallager_type_I}

The following bound is due to Gallager.
\begin{proposition}[Gallager type-I bound~\cite{Gallager63PhDLDPC}, Sec. 3.3]
For any constant $t$:
\begin{align}
\epsilon^{\textrm(RCSE)} \leq \epsilon_\textrm{G-I},
\end{align}
where
\begin{align}
\label{eq:GallagerTypeI}
\epsilon_\textrm{G-I}
&\mDefine \pr{ i(\XX;\YY) < t } + \nonumber
\\&\phantom{\leq} + (M-1)\pr{ i(\XX;\YY) \geq t \land i(\bar{\XX};\YY) \geq i(\XX;\YY) }.
\end{align}
\end{proposition}

Just like the RCU, this bound is based upon a union bound for the ML decoder. However, it is inferior to the RCU bound, due to the following consideration. Taking the minimum between the union and one in the RCU bound is similar to the threshold $t$ in~\eqref{eq:GallagerTypeI}, in the way that it avoids over-estimating the error probability in cases where the channel behavior was ``bad''. However, the RCU bound uses the optimal threshold given $\XX$ and $\YY$; the Gallager bound uses a \emph{global} threshold, which reflects a tradeoff. Nevertheless, for additive channels (recall Definition~\ref{definition:additive}) the local and global thresholds coincide.


\begin{proposition}For any RCSE:
\label{prop:GI_geq_RCU}
\begin{align}
\epsilon_\textrm{G-I} \geq \epsilon_\textrm{RCU},
\end{align}
where $\epsilon_\textrm{G-I}$ and $\epsilon_\textrm{RCU}$ are defined in~\eqref{eq:GallagerTypeI} and in~\eqref{eq:RCU} respectively.
If the channel is additive and the code ensemble is PIE with uniform distribution over $\mset{X}$, then equality holds.
\end{proposition}

\begin{proof}
For the first part, define the events
$A \mDefine \left\{i(\bar{\XX};\YY) \geq i(\XX;\YY) \right\}$ and $T \mDefine \left\{ i(\XX;\YY) \geq t \right\}$ ($T^c$ denotes the complementary event). Then:
\begin{subequations}
\begin{align}
\epsilon_\textrm{RCU} &= \EE_{\XX,\YY} \left[ \min \left\{1,(M-1)\cdot(1-Z) \right\} \right]
\\&= \pr{T^c} \cdot \EE_{\XX,\YY} \left[ \min \left\{1,(M-1)\cdot(1-Z) \right\} \mid T^c \right] \nonumber
\\&\phantom{=} +\pr{T} \cdot \EE_{\XX,\YY} \left[ \min \left\{1,(M-1)\cdot(1-Z) \right\} \mid T \right]
\\&\leq \pr{T^c} +\pr{T} \cdot \EE_{\XX,\YY} \left[ (M-1)\pr{A \mid \XX,\YY} \mid T \right]
\label{eq:RCU_leq_GTI:proof:inequality}
\\&= \pr{T^c} + (M-1)\pr{T} \cdot \pr{A \mid T}
\\&= \epsilon_\textrm{G-I}
\end{align}
\end{subequations}

For the second part, recall that by Lemma~\ref{lemma:additive}, $ i(\bar\XX;\YY)$ is statistically independent of $(\XX,\YY)$.
Denote by $t^*$ the minimal threshold $t$ such that
\begin{align*}
(M-1) \pr{ i(\bar\XX;\YY) \geq t } \leq 1.
\end{align*}
Then
$
(M-1) \pr{ i(\bar\XX;\YY) \geq i(\XX;\YY) \mid i(\XX;\YY)<t^* } \geq 1.
$ 
Under the notation of $U$ from the first part the proof, we have that:
$
\EE_{\XX,\YY} \left[ \min \left\{1,U \right\} \mid i(\XX;\YY)<t^* \right] = 1,
$ 
i.e., the inequality in~\eqref{eq:RCU_leq_GTI:proof:inequality} is equality in this case.
\end{proof}

\begin{remark} It follows, that for the BSC, $\epsilon_\textrm{G-I} = \epsilon_\textrm{RCU}$. Indeed, it is noted in~\cite{PolyanskiyPVFiniteLength10} that for the BSC, the RCU bound is equal to Poltyrev's bound~\cite{Poltyrev:94:BSC}; this is not surprising, since the latter is derived from~\eqref{eq:GallagerTypeI} (Poltyrev's bound uses linear codes, see Section~\ref{sec:linear} in the sequel). \end{remark}

\begin{remark} Gallager's type I bound can be improved by breaking ties, similar to the improvement of RCU$^*$, leading to G-I$^*$. An analysis result to Proposition~\ref{prop:GI_geq_RCU} relates G-I$^*$ and RCU$^*$.\end{remark}

\subsection{Threshold-Decoding Union Bounds}

The average error probability of an RCSE can be further lower-bounded using the sub-optimal  \emph{threshold decoder}~\cite{Feinstein54}. This decoder looks for a codeword that has a likelihood score above some predetermined threshold. In \cite[Theorem 18]{PolyanskiyPVFiniteLength10} a union bound is derived for such a decoder, where if multiple codewords pass the threshold, the winner is chosen uniformly from among them.\footnote{In fact, the proof states that the ``first'' codeword to pass the threshold is selected. However, such ordering of the codewords is not required.} The resulting ``dependence testing'' (DT) bound is given by:
\begin{subequations}
\label{eq:DT}
\begin{align} \epsilon_\textrm{DT} \triangleq \pr{i(\XX;\YY) \leq \gamma} + \frac{M-1}{2} \pr{i(\bar \XX,\YY) > \gamma}, \end{align} where the optimal threshold is given by\footnote{In~\cite{MartinezFabregas2011:ISIT}, the threshold is further optimized, depending on the competing codeword and on the received word}
\begin{align} \label{eq:gamma} \gamma = \log \frac{M-1}{2}. \end{align} \end{subequations}

A troubling behavior, demonstrated in \cite{PolyanskiyPVFiniteLength10} using the binary erasure channel (BEC), is that sometimes $\epsilon_\textrm{RCU}>\epsilon_\textrm{DT}$. This is counter-intuitive since the DT bound is derived from a sub-optimal decoder. We find that this artifact stems from the fact that the RCU bound ignores ties, and prove that the improved bound, denoted by RCU$^*$, always satisfies $\epsilon_{\textrm{RCU}^*}\leq\epsilon_\textrm{DT}$. To that end, we prove a (very slightly) improved bound for the threshold decoder, that is closer in form to the ML bounds~\eqref{eq:RCU} and~\eqref{eq:GallagerTypeI}. It uses the following definitions (cf. \eqref{eq:WZ}).
\begin{subequations} \label{eq:WZ_q}
\begin{align}
W_q & = \pr{ q(i(\bar{\XX};\YY)) = q(i(\XX;\YY)) | \XX, \YY } \\
Z_q & = \pr{ q(i(\bar{\XX};\YY)) < q(i(\XX;\YY)) | \XX, \YY },
\end{align}
where $q(i)$ is the indicator function:
\begin{align} \label{eq:q}
q(i) = \indicator{i > \gamma}.
\end{align}
\end{subequations}

\begin{proposition}
For an RCSE,
\begin{align}
\epsilon^{\textrm{(RCSE)}} \leq \epsilon_{\textrm{TU}},
\end{align}
where
\begin{align}
\label{eq:TU}
\epsilon_{\textrm{TU}} \triangleq \EE_{\XX,\YY} \left[ \min \left\{
1,
(M-1) \cdot \left(1-Z_q-\frac{W_q}{2}\right) \right\} \right].
\end{align}
Furthermore, $\epsilon_\textrm{TU} \leq \epsilon_\textrm{DT}$.
\end{proposition}

\begin{proof}
For proving achievability, consider a decoder identical to the ML decoder, except that before comparing the words, the information-density scores are quantized according to \eqref{eq:q}. For the comparison to the DT bound,
\begin{subequations}
\begin{align}
\epsilon_\textrm{DT}
&= E_{\XX,\YY} \left[ \indicator{i(\XX;\YY) \leq \gamma} + \frac{M-1}{2} \pr{i(\bar \XX,\YY) > \gamma \mid \XX,\YY} \right]
\\& \geq E_{\XX,\YY} \left[ \min \left\{ 1, (M-1) \cdot \left[ \indicator{i(\XX;\YY) \leq \gamma} \phantom{\frac{1}{2}}\right.\right.\right. \nonumber
\\&\phantom{\geq} \left.\left.\left. + \frac{1}{2} \pr{i(\bar \XX,\YY) > \gamma \mid \XX,\YY} \right] \right\} \right]
\\& \geq E_{\XX,\YY} \left[ \min \left\{ 1, \frac{M-1}{2} \cdot \left[ \indicator{i(\XX;\YY) \leq \gamma}  \right.\right.\right. \nonumber
\\&\phantom{\geq}\left.\left.\phantom{\frac{1}{2}}\left. + \pr{i(\bar \XX,\YY) > \gamma \mid \XX,\YY} \right] \right\} \right]
\\&= \epsilon_\textrm{TU}.
\end{align}
\end{subequations}
\end{proof}

\begin{remark} It is not obvious that the optimal threshold for the TU bound is $\gamma$ of \eqref{eq:gamma}. However, it is good enough for our purposes. \end{remark}

\begin{proposition}
For any channel, $\epsilon_{\textrm{RCU}^*}\leq\epsilon_\textrm{TU}$. Thus, the RCU$^*$ bound is tighter than the DT bound, i.e.:
\begin{align*}
\epsilon_{\textrm{RCU}^*}\leq\epsilon_\textrm{DT}.
\end{align*}
\end{proposition}

\begin{proof} Recalling Remark~\ref{remark:pairwise}, the RCU$^*$ bound reflects optimal (ML) pairwise decision. Thus, necessarily the pairwise error probabilities satisfy:
\begin{align}
Z + \frac{W}{2} \geq Z_q + \frac{W_q}{2}.
\end{align}

\end{proof}

\begin{remark} In fact, the case of the BEC, where $\epsilon_{\textrm{RCU}^*}=\epsilon_\textrm{TU}=\min\left(1,\epsilon_\textrm{DT}\right)$ is very special. In the BEC, an impostor cannot have a higher score than the true codeword; if the channel realization is such that the non-erased elements of $\xx$ and $\bar \xx$ are equal, then $i(\bar\xx;\yy)=i(\xx;\yy)$, otherwise $i(\bar\xx;\yy)=-\infty$. Thus,
\begin{align}
\label{eq:Pe_RCU_star}
\epsilon_{\textrm{RCU}^*} = \EE_{\XX,\YY} \left[ \min \left\{
1,
\frac { (M-1) W}{2} \right\} \right].
\end{align}
Let $k$ be the number of non-erased symbols out of the block of $n$, then $W=2^{-k}$. Consequently, $(M-1)W/2>1$ if and only if $i(\xx,\yy)<\gamma$, where $\gamma$ is given by~\eqref{eq:gamma}.
\end{remark}

\subsection{Performance Comparison}

Comparison of the different union bounds is given in Figure~\ref{fig:PIE_bounds}. In particular, the effect of tie-breaking on the bounds is shown by the comparison of the RCU bound~\eqref{eq:RCU} and the RCU$^*$ bound~\eqref{eq:RCU_star}.
Notice that this bound depends on the ensemble. Due to Lemma~\ref{lemma:additive}, the computation of the RCU and RCU$^*$ bounds for PIE becomes simple, hence show the bounds for this ensemble. Since an i.i.d. ensemble is also PIE, the exact error probability for i.i.d. ensemble~\eqref{PPV_exact} is given as a reference.

\begin{figure}
\centering
\includegraphics[width=\columnwidth]{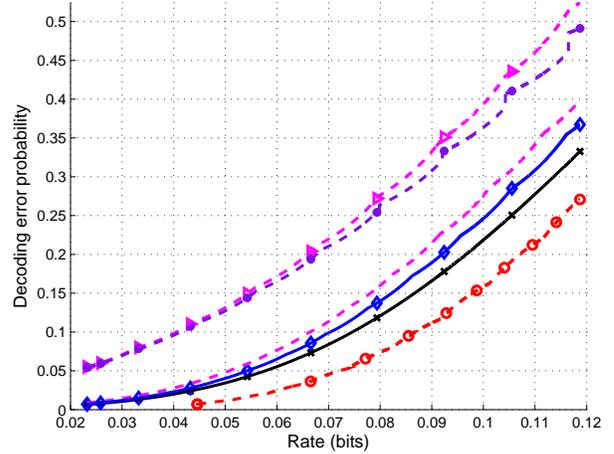} 
\caption{The effect of tie-breaking on the performance of PIE codebooks of the different union-bounds. We demonstrate the effect using a BSC with crossover probability $0.3$, at blocklength $n=100$.
The triangle-marked dashed curve is the DT bound~\eqref{eq:DT}.
The asterisk-marked dashed curve is the TU bound~\eqref{eq:TU}.
The dashed curve is the RCU bound~\eqref{eq:RCU}.
The diamond-marked solid curve is the RCU$^*$ bound~\eqref{eq:RCU_star}.
For reference, we repeat two of the curves of Figure~\ref{fig:iid_bounds}. The $\times$-marked solid curve is the exact performance of the i.i.d. ensemble~\eqref{PPV_exact}, while the circle-marked dashed curve is the lower bound for any codebook~\cite[Theorem 16]{PolyanskiyPVFiniteLength10}. The non-smoothness of some of the curves is not an artifact, but comes from the fact that they involve integers.}
\label{fig:PIE_bounds}
\end{figure}

\section{Linear Codes} \label{sec:linear}

Most known good codes are linear. Beyond that, linear codes have an important role for two reasons. First, they allow to improve performance (both capacity and error probability) in many network problems (see,  e.g.,~\cite{ComputeForward,HaimKochmanErez:2012:Full}). Second, for some channels, the average error probability and maximal error probability coincide for linear codes.

\subsection{The Dithered Linear Codes Ensemble}

For any finite field $\mfield{F}_q$ of cardinality $q$, we define the dithered linear codes ensemble by
\begin{align} \XX_j = H \ww_j + \DD. \end{align}
Here, all operations are defined over the field,
all elements of the $n\times k$ generator matrix $H$ and length-$k$ dither vector $\DD$ are drawn uniformly and independently from the field elements, and $\{\ww_j\}$ are all $k$-vectors over the field. It follows that the codebook size is $M = q^k$. An important special case is when $r$ is prime, and modulo arithmetic is used, e.g., binary (dithered) linear codes.

By~\cite{Dobrushin63}, any dithered linear codes ensemble over this field is PIE. Consequently, the RCU$^*$ bound applies to this ensemble. Further, it is proven in \cite[Appendix A]{PolyanskiyPVFiniteLength10} that for a class of channels, which includes the BSC and the BEC, there exists a \emph{randomized} ML decoder such that the \emph{maximal} error probability $\overline \epsilon$ \eqref{eq:maximal} coincides with the average one. 

\subsection{Additive Channels}

We now restrict our attention to channels that are additive, in the sense of Definition~\ref{definition:additive}. Further, assume that the channels are additive over a finite field, which is the same field over which the code is linear. Clearly, in this situation the dither does not change the distance profile of the codebook, thus it suffices to consider the linear codes ensemble
\begin{align} \XX_j = H \ww_j, \end{align} where again $H$ is i.i.d. uniform over $\mfield{F}_q$. More importantly, in order to achieve good maximal error probability, there is no need to use randomized decoders.


\begin{theorem} For any channel that is additive over a finite field, for an ensemble of linear codes over the field, there exists a deterministic decoder satisfying:
\[\overline \epsilon \leq \epsilon_{\textrm{RCU}^*} \]
\end{theorem}

\begin{remark} Recall that the size of linear code is $M=q^k$ for an integer $k$. Thus, the theorem does not give $\bar \epsilon$ for all $n,M$. \end{remark}

\begin{proof}
Let $\Omega_1,\ldots,\Omega_M$ be a partition of the output space $\mset{Y}^n$ into decision regions (for any $1\leq m \leq M$, $\Omega_m$ is associated with codeword $m$). A partition is optimal in the average error probability sense, if and only if it satisfies: 
\begin{subequations}
\label{eq:Voronoi:general}
\begin{align}
\Omega_m &\subseteq \left\{ \yy \in \mset{Y}^n \mid \forall m' \neq m: V(\yy|\xx_m) \geq V(\yy|\xx_{m'}) \right\}
\\
\Omega_m &\supseteq \left\{ \yy \in \mset{Y}^n \mid \forall m' \neq m: V(\yy|\xx_m) > V(\yy|\xx_{m'}) \right\},
\end{align}
\end{subequations}
and for all $m \neq m'$ $\Omega_m \cap \Omega_{m'} = \emptyset$. 
By~\eqref{eq:infromation_density_of_additive_channel}, we have that for an additive channel,~\eqref{eq:Voronoi:general} is equivalent to
\begin{subequations}
\label{eq:Voronoi:additive_channel}
\begin{align} \label{eq:Voronoi:additive_channel_a}
\Omega_m &\subseteq \left\{ \yy \in \mfield{F}_q^n \mid \forall m' \neq m: P_\NN(\yy-\xx_m) \geq P_\NN(\yy-\xx_{m'}) \right\}. 
\\
\Omega_m &\supseteq \left\{ \yy \in \mfield{F}_q^n \mid \forall m' \neq m: P_\NN(\yy-\xx_m) > P_\NN(\yy-\xx_{m'}) \right\}.
\end{align}
\end{subequations}
Since for any such optimal partition $\epsilon\leq\epsilon_{\textrm{RCU}^*}$, it is sufficient to show that there exists a partition satisfying \eqref{eq:Voronoi:additive_channel} for which $\bar\epsilon=\epsilon$.

Let $\mset{C} \subseteq \mfield{F}_q^n$ be a linear code, and without loss of generality assume that $\xx_1=\mvec{0}$. We construct decoding regions for $\mset{C} $ in the following way. Recall that $\mset{C}$ induces a (unique) partition of the linear space $\mfield{F}_q^n$ into disjoint cosets $\mset{S}_1,\ldots,\mset{S}_{n-k}$, where each coset is a translation of the linear code. For each coset $\mset{S}_j$ let the coset leader $\yy_{j,1}$ be some word that may belong to $\Omega_1$ according to \eqref{eq:Voronoi:additive_channel_a}. By definition, all coset elements of the coset can be uniquely labeled as 
\[\yy_{j,m} = \yy_{j,1} + \xx_m. \] Assign each word $\yy_{j,m}$ to $\Omega_m$. It then satisfies that for all $m'=1,\ldots,M$
\begin{subequations}
\begin{align}
P_\NN(\yy_{j,m}-\xx_m) &= P_\NN(\yy_{j,1})
\\\label{eq:in_proof} &\geq P_\NN(\yy_{j,1}-(\xx_{m'}-\xx_{m}))
\\&= P_\NN(\yy_{j,m}-\xx_{m'}),
\end{align}
\end{subequations}
where in \eqref{eq:in_proof} we have used the facts that $\yy_{j,1}$ satisfies  \eqref{eq:Voronoi:additive_channel_a} and that the sum of codewords is a codeword. It follows, that the partition indeed satisfies \eqref{eq:Voronoi:additive_channel}. To see that $\bar\epsilon=\epsilon$, we have that for all $m=1,\ldots,M$:
\begin{subequations}
\begin{align}
\pr{\YY \in \Omega_m \mid \XX=\xx_m} &= \pr{\XX+\NN \in \Omega_m \mid \XX=\xx_m}
\\&= \pr{\NN \in \Omega_m-\xx_m \mid \XX=\xx_m}
\\\label{eq:in_proof_a} &= \pr{\NN \in \Omega_1 \mid \XX=\xx_m}
\\\label{eq:in_proof_b} &= \pr{\NN \in \Omega_1}
\\&= \pr{\xx_1+\NN \in \Omega_1}
\\&= \pr{\YY \in \Omega_1 \mid \XX=\xx_1},
\end{align}
\end{subequations}
where \eqref{eq:in_proof_a} is due to the construction of $\Omega_m$, and \eqref{eq:in_proof_b} is since the noise is statistically independent of the channel input.
\end{proof}

\begin{remark} The partition used in the proof is not unique, in the sense that for some cosets the choice of the coset leader is arbitrary. However, for any such choice the coset is partitioned in a fair way between the decision regions.
 \end{remark}

\bibliographystyle{IEEEtran}
\bibliography{elih}

\end{document}